\documentclass[runningheads,envcountsame]{llncs}
\usepackage{graphicx}

\usepackage{amssymb}
\usepackage{stmaryrd}  
\usepackage{mathtools} 
\usepackage{enumitem}
\usepackage{hyperref,url}
\allowdisplaybreaks
\spnewtheorem*{remark*}{Remark}{\bfseries}{\rmfamily}
\spnewtheorem*{claim*}{Claim}{\bfseries}{\rmfamily} 
\spnewtheorem*{claimproof}{Proof}{\bfseries}{\rmfamily}

\spnewtheorem*{fact}{Fact}{\it}{\rmfamily}

\newcommand\qedhere{\tag*{\qed}}

\usepackage{tikz}
\usepackage{tikz-cd}
\usetikzlibrary{calc,arrows}
\usetikzlibrary{decorations.pathmorphing}

\newcommand{\Lh}{{\mathbf{\Gamma}}} 
\newcommand{\Lo}{\mathbf{L}} 
\newcommand{\mm}{^{-}} 
\newcommand{\cc}{^\circ} 

\newcommand{\domM}{\ensuremath{\mathrm{\dom(-)}}}  
\newcommand{\domP}{\ensuremath{\mathrm{\dom(+)}}}  
\newcommand{\dom}{\ensuremath{\mathrm{dom}}} 
\newcommand{\supp}{\ensuremath{\mathrm{supp}}} 

\newcommand{\upset}{\ensuremath{\mathord{\uparrow}\mkern1mu}} 
\newcommand{\downset}{\ensuremath{\mathord{\downarrow}\mkern1mu}} 
\newcommand{\mee}{\wedge}

\newcommand\ARG{\mkern1.5mu\text{-}\mkern1.5mu}
\newcommand{\ev}{\mathrm{ev}}

\newcommand{\inv}{^{-1}} %
\newcommand\sue{\subseteq} %
\newcommand{\N}{\ensuremath{\mathbb{N}}}
\newcommand{\Z}{\ensuremath{\mathbb{Z}}}
\newcommand{\Q}{\ensuremath{\mathbb{Q}}}
\newcommand{\ui}{[0,1]}
\newcommand{\w}{\widehat}

\newcommand{\mip}{-} 
\newcommand{\miss}{\ensuremath{\sim}} 

\newcommand\ic{\iota}

\newcommand\MBU[2]{\ensuremath{\sem[#1 \geq #2]}} 
\newcommand\MBD[2]{\ensuremath{\sem[#1 < #2]}}    

\newcommand\TOCITE[1][]{[{\color{blue}??}]}

\newcommand\NOdM{Ne\v{s}et\v{r}il and Ossona de Mendez}

\newcommand\ee[1]{\enspace #1 \enspace}
\newcommand\ete[1]{{\enspace\text{#1}\enspace}}
\newcommand\qtq[1]{{\quad\text{#1}\quad}}

\newcommand\FO{\mathrm{FO}}
\newcommand\Fin{\mathrm{Fin}}
\newcommand\Fs{{\mathcal{F}}}
\newcommand\SP[1]{\left< #1 \right>_{\mathrm{\Gamma}}}
\newcommand\SPc[1]{\left< #1 \right>_{\mathrm{I}}}
\newcommand\PrG[1]{\mathbb{P}_{\geq #1}\,}
\newcommand\PrL[1]{\mathbb{P}_{< #1}\,}
\newcommand\true{\mathbf{t}}
\newcommand\false{\mathbf{f}}

\newcommand{\sem}[1][{\ARG}]{\ensuremath{\llbracket #1 \rrbracket}}
\newcommand{\Pf}{F} 

\newcommand\Formn{\mathrm{Fm}_n}
\newcommand\Modn{\mathrm{Mod}_n}
\DeclareMathOperator{\Typ}{Typ}

\renewcommand{\P}{\raisebox{.17\baselineskip}{\Large\ensuremath{\wp}}} 

\newcommand{\op}{\partial} 
\DeclareMathOperator{\Mea}{\mathcal{M}_{\Gamma}}
\DeclareMathOperator{\Meac}{\mathcal{M}_I}
\newcommand\Shat{\w{\mathbf S}}
\DeclareMathOperator{\B}{\mathcal{B}} 

\DeclareMathOperator{\PFG}{\mathbf P} 
\renewcommand{\restriction}{\mathord{\upharpoonright}}
\newcommand{\blank}{(\ARG)}
\newcommand{\LG}{\gamma^{\#}}    
\newcommand{\LC}{\ic^{\#}}    
\newcommand\PL[1]{P\mkern-1mu\mathcal{L}_{#1}}

\newcommand{\mono}{\hookrightarrow} 
\newcommand{\epi}{\twoheadrightarrow} 
\renewcommand{\epsilon}{\varepsilon}
\renewcommand{\theta}{\vartheta}
\renewcommand{\phi}{\varphi}

\begin{document}

\setlength{\abovedisplayskip}{6pt}
\setlength{\belowdisplayskip}{6pt}
\setlength{\belowcaptionskip}{-20pt}

\title{A duality theoretic view on limits of finite structures\thanks{This project has been supported by the European Research Council (ERC) under the European Union's Horizon 2020 research and innovation program (grant agreement No.670624). Luca Reggio has received an individual support under the grants GA17-04630S of the Czech Science Foundation, and  No.184693 of the Swiss National Science Foundation.}}

\author{Mai Gehrke\inst{1} \and Tom\'{a}\v{s} Jakl\inst{1} \and Luca Reggio\inst{2}}
\authorrunning{M. Gehrke et al.}
\institute{CNRS and Universit\'e C{\^o}te d'Azur, France \\
\email{\{mgehrke,tomas.jakl\}@unice.fr}
\and
Institute of Computer Science of the Czech Academy of Sciences, Czech Republic and 
Mathematical Institute, University of Bern, Switzerland\\
\email{luca.reggio@math.unibe.ch}}

\maketitle              

\begin{abstract}
    A systematic theory of \emph{structural limits} for finite models has been developed by Ne{\v s}et{\v r}il and Ossona de Mendez. It is based on the insight that the collection of finite structures can be embedded, via a map they call the \emph{Stone pairing}, in a space of measures, where the desired limits can be computed. We show that a closely related but finer grained space of measures arises --- via Stone-Priestley duality and the notion of types from model theory --- by enriching the expressive power of first-order logic with certain ``probabilistic operators''.  We provide a sound and complete calculus for this extended logic and expose the functorial nature of this construction.
    
The consequences are two-fold. On the one hand, we identify the logical gist of the theory of structural limits. On the other hand, our construction shows that the duality-theoretic variant of the Stone pairing captures the adding of a layer of quantifiers, thus making a strong link to recent work on semiring quantifiers in logic on words. In the process, we identify the model theoretic notion of \emph{types} as the unifying concept behind this link. These results contribute to bridging the strands of logic in computer science which focus on semantics and on more algorithmic and complexity related areas, respectively.

\end{abstract}

\section{Introduction}

While topology plays an important role, via Stone duality, in many parts of semantics, topological methods in more algorithmic and complexity oriented areas of theoretical computer science are not so common. One of the few examples,
 the one we want to consider here, is the study of limits of finite relational structures. 
We will focus on the \emph{structural limits} introduced by Ne{\v s}et{\v r}il and Ossona de Mendez~\cite{nevsetril2012model,NO2017}. These provide a common generalisation of various notions of limits of finite structures studied in probability theory, random graphs, structural graph theory, and finite model theory. 
The basic construction in this work is the so-called \emph{Stone pairing}. 
Given a relational signature $\sigma$ and a first-order formula $\phi$ in the signature $\sigma$ with free variables $v_1, \dots, v_n$, define
\begin{equation}\label{eq:Stone-Pairing}
\left<\phi, A\right> = \frac{|\{ \overline a \in A^n \mid A \models \phi(\overline a)\}|}{|A|^n} \qquad \parbox{13em}{\centering\emph{(the probability that a random assignment in $A$ satisfies $\phi$)}.} 
\end{equation}
Ne{\v s}et{\v r}il and Ossona de Mendez view the map $A \mapsto \left<\ARG, A\right>$ as an embedding of the finite $\sigma$-structures into the space of probability measures over the Stone space dual to the Lindenbaum-Tarski algebra of all first-order formulas in the signature $\sigma$. This space is complete and thus provides the desired limit objects for all sequences of finite structures which embed as Cauchy sequences.

Another example of topological methods in an algorithmically oriented area of computer science is the use of profinite monoids in automata theory. In this setting, profinite monoids are the subject of the extensive theory, based on theorems by Eilenberg and Reiterman, and used, among others, to settle decidability questions~\cite{Pin09}.
In~\cite{GGP2008}, it was shown that this theory may be understood as an application of Stone duality, thus making a bridge between semantics and more algorithmically oriented work. 
Bridging this semantics-versus-algorithmics gap in theoretical computer science has since gained quite some momentum, notably with the recent strand of research by Abramsky, Dawar and co-workers~\cite{Abramsky2017b,AbramskyShah2018}. 
In this spirit, a natural question is whether the structural limits of \NOdM{} also can be understood semantically, and in particular whether the topological component may be seen as an application of Stone duality. 

 More precisely, recent work on understanding quantifiers in the setting of languages over finite words~\cite{GPR2017} has shown that adding a layer of certain quantifiers (such as classical and modular quantifiers) corresponds dually to measure space constructions. The measures involved are not classical but only finitely additive and they take values in finite semirings rather than in the unit interval. Nevertheless, this appearance of \emph{measures as duals of quantifiers} begs the further question whether the measure spaces in the theory of structural limits may be obtained via Stone duality from a semantic addition of certain quantifiers to classical first-order logic.

The purpose of this paper is to address this question. Our main result is that the Stone pairing of \NOdM{} 
is related by a retraction to a Stone space of measures, which is dual to the Lindenbaum-Tarski algebra of a logic fragment obtained from first-order logic by adding one layer of a probabilistic quantifier, and which arises in exactly the same way as the spaces of semiring-valued measures in logic on words. That is, the Stone pairing, although originating from other considerations, may be seen as arising by duality from a semantic construction.

\smallskip

A foreseeable hurdle is that spaces of classical measures are valued in the unit interval $\ui{}$ which is not zero-dimensional and hence outside the scope of Stone duality. This is well-known to cause problems e.g.\ in attempts to combine non-determinism and probability  in domain theory~\cite{Jung13}. However, in the structural limits of \NOdM, at the base, one only needs to talk about finite models equipped with normal distributions and thus only the finite intervals $I_n = \{ 0, \frac{1}{n}, \frac{2}{n}, \dots, 1\}$ are involved. A careful duality-theoretic analysis identifies a codirected diagram (i.e.\ an inverse limit system) based on these intervals compatible with the Stone pairing. The resulting inverse limit, which we denote $\Lh$, is a Priestley space. It comes equipped with an algebra-like structure, which allows us to reformulate many aspects of the theory of structural limits in terms of $\Lh$-valued measures as opposed to $\ui$-valued measures.

The analysis justifying the structure of $\Lh$ is  based on duality theory for double quasi-operator algebras~\cite{GP2007,GP2007b}. In the presentation, we have tried to compromise between giving interesting topo-relational insights into why $\Lh$ is as it is, and not overburdening the reader with technical details.
Some interesting features of $\Lh$, dictated by the nature of the Stone pairing and the ensuing codirected diagram, are that
\begin{itemize}
\item $\Lh$ is based on a version of $\ui$ in which the rationals are doubled;
\item $\Lh$ comes with section-retraction maps $\begin{tikzcd}[cramped,sep=1.8em] \ui \rar[hook]{\ic} & \Lh \rar[->>]{\gamma} & \ui\end{tikzcd}$;
\item the map $\ic$ is lower semicontinuous while the map $\gamma$ is continuous.
\end{itemize}
These features are a consequence of general theory and precisely allow us to witness continuous phenomena relative to $\ui$ in the setting of $\Lh$. 

\subsection*{Our contribution}

We show that the ambient measure space for the structural limits of \NOdM{} can be obtained via \emph{``adding a layer of quantifiers''} in a suitable enrichment of first-order logic. The conceptual framework for seeing this is that of \emph{types} from classical model theory.
 More precisely, we will see that a variant of the Stone pairing is a map into a space of measures with values in a Priestley space $\Lh$. Further, we show that this map is in fact the embedding of the finite structures into the space of ($0$-)types of an extension of first-order logic, which we axiomatise.
On the other hand, $\Lh$-valued measures and $\ui$-valued measures are tightly related by a retraction-section pair which allows the transfer of properties.
These results identify the logical gist of the theory of structural limits and provide a new interesting connection between logic on words and the theory of structural limits in finite model theory.

\paragraph*{Outline of the paper.}
In section \ref{s:prel} we briefly recall Stone-Priestley duality, its application in logic via spaces of types, and the particular instance of logic on words (needed only to show the similarity of the constructions). In Section \ref{s:Lh} we introduce the Priestley space $\Lh$ with its additional operations, and show that it admits $\ui$ as a retract. The spaces of $\Lh$-valued measures are introduced in Section \ref{s:spaces-of-measures}, and the retraction of $\Lh$ onto $\ui$ is lifted to the appropriate spaces of measures. In Section \ref{s:stone-pairing} we introduce the $\Lh$-valued Stone pairing and make the link with logic on words. Further, we compare convergence in the space of $\Lh$-valued measures with the one considered by Ne{\v s}et{\v r}il and Ossona de Mendez. Finally, in Section \ref{s:logic-of-measures} we show that constructing the space of $\Lh$-valued measures dually corresponds to enriching the logic with probabilistic operators.
\section{Preliminaries}\label{s:prel}
\paragraph*{Notation.} Throughout this paper, if $X\xrightarrow{f} Y\xrightarrow{g} Z$ are functions, their composition is denoted $g\cdot f$. For a subset $S\subseteq X$, $f_{\restriction S}\colon S\to Y$ is the obvious restriction. Given any set $T$, $\P(T)$ denotes its power-set. Further, for a poset $P$, $P^\op$ is the poset obtained by turning the order of $P$ upside down.
\subsection{Stone-Priestley duality}\label{s:duality}
In this paper, we will need Stone duality for bounded distributive lattices in the order topological form due to Priestley \cite{Priestley1970}. It is a powerful and well established tool in the study of propositional logic and semantics of programming languages, see e.g. \cite{Goldblatt1989,Abramsky91} for major landmarks. We briefly recall how this duality works.

A \emph{compact ordered space} is a pair $(X,\leq)$ where $X$ is a compact space and $\leq$ is a partial order on $X$ which is closed in the product topology of $X\times X$. (Note that such a space is automatically Hausdorff). A compact ordered space is a \emph{Priestley space} provided it is \emph{totally order-disconnected}. That is, for all $x,y\in X$ such that $x\not\leq y$, there is a \emph{clopen} (i.e.\ simultaneously closed and open) $C\subseteq X$ which is an up-set for $\leq$, and satisfies $x\in C$ but $y\notin C$. We recall the construction of the Priestley space of a distributive lattice $D$.\footnote{We assume all distributive lattices are bounded, with the bottom and top denoted by $0$ and $1$, respectively. The bounds need to be preserved by homomorphisms.}

 A non-empty proper subset $F\subset D$ is a \emph{prime filter} if it is \emph{(i)} upward closed (in the natural order of $D$), 
 \emph{(ii)} closed under finite meets, and \emph{(iii)} if $a\vee b\in F$, either $a\in F$ or $b\in F$. Denote by $X_D$ the set of all prime filters of $D$. By Stone's Prime Filter Theorem, the map
\begin{align*}
\sem\colon D\to \P(X_D), \ \ a\mapsto \sem[a]=\{F\in X_D\mid a\in F\}
\end{align*}
is an embedding.
Priestley's insight was that $D$ can be recovered from $X_D$, if the latter is equipped with the inclusion order and the topology generated by the sets of the form $\sem[a]$ and their complements. This makes $X_D$ into a Priestley space --- the \emph{dual space} of $D$ --- and the map $\sem$ is an isomorphism between $D$ and the lattice of clopen up-sets of $X_D$. Conversely, any Priestley space $X$ is the dual space of the lattice of its clopen up-sets. We call the latter the \emph{dual lattice} of $X$. This correspondence extends to morphisms. In fact, Priestley duality states that the category of distributive lattices with homomorphisms is dually equivalent to the category of Priestley spaces and continuous monotone maps.
When restricting to Boolean algebras, we recover the celebrated Stone duality restricted to Boolean algebras and \emph{Boolean spaces}, i.e.\ compact Hausdorff spaces in which the clopen subsets form a basis. 

\subsection{Stone duality and logic: type spaces}
The \emph{theory of types} is an important tool for first-order logic. We briefly recall the concept as it is closely related to, and provides the link between, two otherwise unrelated occurrences of topological methods in theoretical computer science.

Consider a signature $\sigma$ and a first-order theory $T$ in this signature. For each $n\in\N$, let $\Formn$ denote the set of first-order formulas whose free variables are among  $\overline{v}=\{v_1,\dots,v_n\}$, and let $\Modn(T)$ denote the class of all pairs $(A,\alpha)$ where $A$ is a model of $T$ and $\alpha$ is an interpretation of $\overline{v}$ in $A$. Then the satisfaction relation, $(A,\alpha)\models \varphi$, is a binary relation from $\Modn$ to $\Formn$. It induces the equivalence relations of elementary equivalence $\equiv$ and logical equivalence $\approx$ on these sets, respectively. The quotient $\FO_n(T)=\Formn/{\approx}$ carries a natural Boolean algebra structure and is known as the \emph{$n$-th Lindenbaum-Tarski algebra} of $T$. Its dual space is $\Typ_n(T)$, the \emph{space of $n$-types} of $T$, whose points can be identified with elements of $\Modn(T)/{\equiv}$. The Boolean algebra $\FO(T)$ of \emph{all} first-order formulas modulo logical equivalence over $T$ is the directed colimit of the $\FO_n(T)$ for $n\in\N$ while its dual space, $\Typ(T)$, is the codirected limit of the $\Typ_n(T)$ for $n\in\N$ and consists of models equipped with interpretations of the full set of~variables.

If we want to study finite models, there are two equivalent approaches: e.g.\ at the level of sentences, we can either consider the theory $T_\text{\em fin}$ of finite $T$-models, or the closure of the collection of all finite $T$-models in the space $\Typ_0(T)$. This closure yields a space, which should tell us about finite $T$-structures. Indeed, it is equal to $\Typ_0(T_\text{\em fin})$, the space of pseudofinite $T$-structures. For an application of this, see~\cite{vanGoolSteinberg}.
Below, we will see an application in finite model theory of the case $T=\emptyset$ (in this case we write $\FO(\sigma)$ and $\Typ(\sigma)$ instead of $\FO(\emptyset)$ and $\Typ(\emptyset)$ to at least flag the signature).

\smallskip
In light of the theory of types as exposed above, the Stone pairing of \NOdM{} (see equation~\eqref{eq:Stone-Pairing}) can be regarded as an embedding of finite structures into the space of probability measures on $\Typ(\sigma)$, which set-theoretically are finitely additive functions $\FO(\sigma) \to \ui$.

\subsection{Duality and logic on words}\label{s:duality-and-low}
As mentioned in the introduction, spaces of measures arise via duality in \emph{logic on words}~\cite{GPR2017}.  Logic on words, as introduced by B\"uchi, see e.g.~\cite{MaSch08} for a recent survey, is a variation and specialisation of finite model theory where only models based on words are considered. I.e., a word $w\in A^*$ is seen as a relational structure on $\{1,\ldots, |w|\}$, where $|w|$ is the length of $w$, equipped with a unary relation $P_a$, for each $a\in A$, singling out the positions in the word where the letter $a$ appears. Each sentence $\phi$ in a language interpretable over these structures yields a language $L_{\phi}\subseteq A^*$ consisting of the words satisfying $\phi$. Thus, logic fragments are considered modulo the theory of finite words and the Lindenbaum-Tarski algebras are subalgebras of $\P(A^*)$ consisting of the appropriate $L_{\phi}$'s, cf.~\cite{vanGoolSteinberg} for a treatment of first-order logic on words. 
 
For lack of logical completeness, the duals of the Lindenbaum-Tarski algebras have more points than those given by models. Nevertheless, the dual spaces of types, which act as compactifications and completions of the collections of models, provide a powerful tool for studying logic fragments by topological means. The central notion is that of \emph{recognition}, in which, a Boolean subalgebra $\B\subseteq\P(A^*)$ is studied by means of the dual map $\eta\colon\beta(A^*)\to X_{\B}$. Here $\beta(A^*)$ is the Stone dual of $\P(A^*)$, also known in topology as the {\v C}ech-Stone compactification of the discrete space $A^*$, and $X_{\B}$ is the Stone dual of $\B$. The set $A^*$ embeds in $\beta(A^*)$, and $\eta$ is uniquely determined by its restriction $\eta_0\colon A^*\to X_{\B}$. Now, Stone duality implies that $L\subseteq A^*$ is in $\B$ iff there is a clopen subset $V\subseteq X_{\B}$ so that 
$\eta_0^{-1}(V)=L$. Anytime the latter is true for a map $\eta$ and a language $L$ as above, one says that \emph{$\eta$ recognises $L$}.\footnote{Here, being beyond the scope of this paper, we are ignoring the important role of the monoid structure available on the spaces (in the form of profinite monoids or BiMs, cf.\ \cite{vanGoolSteinberg,GPR2017}).}

When studying logic fragments via recognition, the following inductive step is central: given a notion of quantifier and a recogniser for a Boolean algebra of formulas with a free variable, construct a recogniser for the Boolean algebra generated by the formulas obtained by applying the quantifier. 
 This problem was solved in \cite{GPR2017}, using duality theory, in a general setting of \emph{semiring quantifiers}. The latter are defined as follows: let $(S,+,\cdot,0_S,1_S)$ be a semiring, and $k\in S$. Given a formula $\psi(v)$, the formula $\exists_{S,k}v.\psi(v)$ is true of a word $w\in A^*$ iff $k = 1_S+\cdots+ 1_S$, $m$ times,
where $m$ is the number of assignments of the variable $v$ in $w$ satisfying $\psi(v)$. 
If $S=\Z/q\Z$, we obtain the so-called \emph{modular quantifiers}, and for $S$ the two-element lattice we recover the existential quantifier $\exists$. 

To deal with formulas with a free variable, one considers maps of the form $f\colon \beta((A\times 2)^*)\to X$ (the extra bit in $A\times 2$ is used to mark the interpretation of the free variable). In \cite{GPR2017} (see also \cite{GPR2019}), it was shown that $L_{\psi(v)}$ is recognised by $f$ iff for every $k\in S$ the language $L_{\exists_{S,k}v.\psi(v)}$ is recognised by the composite
\begin{equation} \label{def:R}
\xi\colon A^* \xrightarrow{\mathmakebox[3em]{R}} \Shat(\beta((A\times 2)^*)) \xrightarrow{\mathmakebox[3em]{\Shat(f)}} \Shat(X),
\end{equation}
 where $\Shat(X)$ is the space of finitely additive $S$-valued measures on $X$
and $R$ maps $w\in A^*$ to the measure $\mu_w\colon \P((A\times 2)^*) \to S$ sending a set $K\subseteq (A\times 2)^*$ to the sum $1_S + \cdots + 1_S$, $n_{w,K}$ times. Here, $n_{w,K}$ is the number of interpretations $\alpha$ of the free variable $v$ in $w$ such that the pair $(w,\alpha)$, seen as an element of $(A\times 2)^*$, belongs to $K$. Finally, $\Shat(f)$ sends a measure to its pushforward along $f$.

\section{The space $\Lh$}\label{s:Lh}
Central to our results is a Priestley space $\Lh$ closely related to $\ui$, in which our measures will take values. 
Its construction comes from the insight that the range of the Stone pairing $\left<\ARG, A\right>$, for a finite structure $A$ and formulas restricted to a fixed number of free variables, can be confined to a chain $I_n=\{ 0, \frac{1}{n}, \frac{2}{n}, \dots, 1\}$. Moreover, the floor functions $f_{mn,n}\colon I_{mn}\twoheadrightarrow I_n$ are monotone surjections. The ensuing system $\{f_{mn,n}\colon I_{mn}\twoheadrightarrow I_n\mid m,n\in \N\}$ can thus be seen as a codirected diagram of finite discrete posets and monotone maps. Let us define $\Lh$ to be the limit of this diagram. Then, $\Lh$ is naturally equipped with a structure of Priestley space, see e.g.\ \cite[Corollary VI.3.3]{Johnstone1986}, and can be represented as based on the set 
\[\{ r\mm \mid r\in (0,1]\} \cup \{ q\cc \mid q\in \Q\cap\ui \}.
\] 
The order of $\Lh$ is the unique total order which has $0\cc$ as bottom element, satisfies $r^*< s^*$ if and only if $r<s$ for ${\scriptstyle \ast}\in\{{\scriptstyle-},{\scriptstyle\circ}\}$, and such that $q\cc$ is a cover of $q\mm$ for every rational $q\in(0,1]$ (i.e.\ $q\mm<q\cc$, and there is no element strictly in between). In a sense, the values $q\mm$ represent approximations of the values of the form $q\cc$. Cf.~Figure~\ref{f:Lo-Lh}.
The topology of $\Lh$ is generated by the sets of the form
\[ \upset p\cc = \{x\in\Lh \mid p\cc\leq x\}
   \qtq{and}
   \downset q\mm=\{x\in\Lh \mid x\leq q\mm\}
\]
for $p,q\in \Q\cap\ui$ such that $q\neq 0$.
The distributive lattice dual to $\Lh$, denoted by $\Lo$, is given by
\begin{equation*}
    \Lo=  \{ \bot\} \cup (\Q\cap\ui)^\op, \text{ with } \bot <_\Lo q \text{ and } q \leq_\Lo p \text{ for every } p \leq q \text{ in } \Q\cap\ui.
\end{equation*}
\vspace{-2.0em}
\begin{figure}[htb]
\centering
\begin{tikzpicture}
    \begin{scope}
        \node at (0,0) (bot) {};
        \node at (0,0.15) (1) {};
        \node at (0,2) (0) {};

        \draw[densely dotted] (1.center) -- (0.center);

        \node at ($(bot)+(-0.35,0)$) {$\bot$};
        \node at ($(1)+(0.3,0.0)$) {1};
        \node at ($(0)+(0.3,0)$) {0};

        \foreach \pt in {1,0,bot} {
            \draw ($(\pt)-(0.1,0)$) -- ($(\pt)+(0.1,0)$);
        }

        \node at (-1, 1) {$\Lo =$};

        \draw [<->,
        line join=round,
        decorate,
        decoration={
            zigzag,
            segment length=5,
            amplitude=1,
            post=lineto,
            post length=4pt,
            pre length=4pt
        }]  (-2,1) to (-4,1);
    \end{scope}

    \begin{scope}[xshift=-17em]
        \node at (0,2.15) (1cc) {};
        \node at (0,2) (1mm) {};
        \node at (0,0) (0cc) {};
        \node at (0,1.4) (r) {};
        \node at (0,0.95) (qc) {};
        \node at (0,0.80) (qm) {};
        \node at ($(r) -(0.4,0.0)$) (rmm) {$r\mm$};
        \node at ($(qc)+(0.4,0.05)$) (qcc) {$q\cc$};
        \node at ($(qm)-(0.4,0.05)$) (qmm) {$q\mm$};
        \node at ($(1cc)+(0.4,0)$) {$1\cc$};
        \node at ($(1mm)-(0.4,0.05)$) {$1\mm$};
        \node at ($(0cc)+(0.4,0)$) {$0\cc$};

        \draw[densely dotted] (1mm.center) -- (qc.center);
        \draw[densely dotted] (qm.center) -- (0cc.center);

        \foreach \pt in {1cc,1mm,0cc,r,qm,qc} {
            \draw ($(\pt)-(0.1,0)$) -- ($(\pt)+(0.1,0)$);
        }
        \node at (-1.3, 1) {$\Lh =$};
    \end{scope}
\end{tikzpicture}
\caption{The Priestley space $\Lh$ and its dual lattice $\Lo$}
\label{f:Lo-Lh}
\end{figure}

\subsection{The algebraic structure on $\Lh$}\label{subs:mip-and-miss}
When defining measures we need an algebraic structure available on the space of values. The space $\Lh$ fulfils this requirement as it comes equipped with a partial operation $\mip\colon \domM\to \Lh$, where $\domM = \{(x,y)\in \Lh\times \Lh\mid y\leq x\}$ and
\[
    \begin{aligned}
        r\cc \mip s\cc &\ee= (r - s)\cc \\
        r\mm \mip s\cc &\ee= (r-s)\mm
    \end{aligned}
    \qquad\qquad
    \begin{rcases*}
        r\cc \mip s\mm \\
        r\mm \mip s\mm
    \end{rcases*}
    \ee=
    \begin{cases}
        (r - s)\cc & \text{ if } r-s \in \Q \\
        (r - s)\mm & \text{ otherwise}.
    \end{cases}
\]

In fact, this (partial) operation is dual to the truncated addition on the lattice $\Lo$. However, explaining this would require us to delve into extended Priestley duality for lattices with operations, which is beyond the scope of this paper. See \cite{Goldblatt1989} and also \cite{GP2007,GP2007b} for details.
It also follows from the general theory that there exists another partial operation definable from $\mip$, namely:
\begin{equation}\label{eq:def-of-miss}
\miss\colon\domM\to \Lh, \ \ x \miss y = \bigvee{\{ x \mip q\cc \mid y < q\cc \leq x\}}.
\end{equation}
Now we collect some basic properties of $\mip$ and $\miss$, needed in Section~\ref{s:spaces-of-measures}, which follow from the general theory of  \cite{GP2007,GP2007b}.
First, recall that a map into an ordered topological space is \emph{lower} (resp.\ \emph{upper}) \emph{semicontinuous} provided the preimage of any open down-set (resp.\ open up-set) is open.
\begin{lemma}\label{l:properties-of-mip}\label{l:properties-of-mis}
If $\domM$ is seen as a subspace of $\Lh \times \Lh{}^\op$, the following hold:
    \begin{enumerate}
        \item $\domM$ is a closed up-set in $\Lh \times \Lh{}^\op$;
        \item both $\mip\colon \domM\to \Lh$ and $\miss\colon \domM\to \Lh$ are monotone in the first coordinate, and antitone in the second;
        \item $\mip\colon \domM\to \Lh$ is lower semicontinuous;
        \item $\miss\colon \domM\to \Lh$ is upper semicontinuous.
    \end{enumerate}
\end{lemma}

\subsection{The retraction $\Lh\epi \ui$}\label{s:retraction-Lh-ui}
In this section we show that, with respect to appropriate topologies, the unit interval $\ui$ can be obtained as a topological retract of $\Lh$, in a way which is compatible with the operation $\mip$. This will be important in Sections~\ref{s:spaces-of-measures} and~\ref{s:stone-pairing}, where we need to move between \ui-valued and $\Lh$-valued measures.
Let us define the monotone surjection given by collapsing the doubled elements:
\begin{equation}\label{eq:gamma}
\gamma\colon \Lh \to \ui, \ r\mm, r\cc \mapsto r.
\end{equation}
The map $\gamma$ has a right adjoint, given by
\begin{equation}\label{eq:ic}
\ic\colon \ui \to \Lh, \ r \mapsto
        \begin{cases}
            r\cc & \text{if } r\in \Q \\
            r\mm & \text{otherwise}.
        \end{cases}
\end{equation}
Indeed, it is readily seen that $\gamma(y)\leq x$ iff $y\leq \ic(x)$, for all $y\in \Lh$ and $x\in \ui$. The composition $\gamma\cdot\ic$ coincides with the identity on $\ui$, i.e.\ $\ic$ is a section of $\gamma$. Moreover, this retraction lifts to a topological retract provided we equip $\Lh$ and $\ui$ with the topologies consisting of the open down-sets:

\begin{lemma}\label{l:gamma-ic-p-morphisms}
    The map $\gamma\colon \Lh \to \ui$ is continuous and the map $\ic\colon \ui \to \Lh$ is lower semicontinuous.
\end{lemma}
\begin{proof}
    To check continuity of $\gamma$ observe that, for a rational $q\in (0,1)$, $\gamma\inv(q,1]$ and $\gamma\inv[0,q)$ coincide, respectively, with the open sets
    \[
         \bigcup \{ \upset p\cc \mid  p\in \Q\cap\ui \text{ and } q < p \} \ete{and}
         \bigcup \{ \downset p\mm \mid  p\in \Q\cap (0,1] \text{ and } p < q \}.
     \]
    Also, $\ic$ is lower semicontinuous, for $\ic\inv(\downset q\mm) = [0, q)$ whenever $q\in \Q\cap(0,1]$.\qed
\end{proof}

It is easy to see that both $\gamma$ and $\ic$ preserve the minus structure available on $\Lh$ and \ui{} (the unit interval is equipped with the usual minus operation $x - y$ defined whenever $y\leq x$), that is,
\begin{itemize}
    \item  $\gamma(x \mip y) = \gamma(x \miss y) = \gamma(x) - \gamma(y)$ whenever $y \leq x$ in $\Lh$, and
    \item  $\ic(x - y) = \ic(x) \mip \ic(y)$ whenever $y \leq x$ in \ui{}.
\end{itemize}

\begin{remark*}
$\ic\colon \ui\to\Lh$ is not upper semicontinuous because, for every $q\in\Q\cap\ui$,
$\ic\inv(\upset q\cc)=\{x\in\ui\mid q\cc \leq \ic(x)\}=\{x\in \ui\mid \gamma(q\cc)\leq x\}=[q,1]$.
\end{remark*}

\section{Spaces of measures valued in $\Lh$ and in $\ui$}\label{s:spaces-of-measures}
The aim of this section is to replace $\ui$-valued measures  by $\Lh$-valued measures. The reason for doing this is two-fold. First, the space of $\Lh$-valued measures is Priestley (Proposition \ref{p:measures-Priestley}), and thus amenable to a duality theoretic treatment and a dual logic interpretation (cf.\ Section \ref{s:logic-of-measures}). Second, it retains more topological information than the space of $\ui$-valued measures. Indeed, the former retracts onto the latter (Theorem \ref{th:retraction-measures}).

Let $D$ be a distributive lattice. Recall that, classically, a monotone function $m\colon D\to \ui$ is a (finitely additive, probability) measure provided $m(0) = 0$, $m(1) = 1$, and $m(a) + m(b)=m(a\vee b) + m(a\wedge b)$ for every $a, b\in D$. The latter property is equivalently expressed as
\begin{align}
    \forall a,b\in D, \ m(a)-m(a\wedge b)=m(a\vee b)-m(b).
    \label{e:finite-additive-with-minuses}
\end{align}
We write $\Meac(D)$ for the set of all measures $D\to\ui$, and regard it as an ordered topological space, with the structure induced by the product order and product topology of $\ui^D$. The notion of (finitely additive, probability) $\Lh$-valued measure is analogous to the classical one, except that the finite additivity property \eqref{e:finite-additive-with-minuses} splits into two conditions, involving $\mip$ and $\miss$.
\begin{definition}\label{d:measure}
    Let $D$ be a distributive lattice. A \emph{$\Lh$-valued measure} (or simply a \emph{measure}) on $D$ is a function $\mu\colon D\to \Lh$ such that 
    \begin{enumerate}
        \item $\mu(0) = 0\cc$ and $\mu(1) = 1\cc$,
        \item $\mu$ is monotone, and
        \item for all $a,b\in D$,
            \[ \mu(a)\miss \mu(a\wedge b) \leq \mu(a\vee b) \mip \mu(b) \ete{ and } \mu(a) \mip \mu(a\wedge b) \geq \mu(a\vee b) \miss \mu(b).\]
    \end{enumerate}
    We denote by $\Mea(D)$ the subspace of $\Lh^D$ consisting of the measures $\mu\colon D\to \Lh$.
\end{definition}

Since $\Lh$ is a Priestley space, so is $\Lh^D$ equipped with the product order and topology. Hence, we regard $\Mea(D)$ as an ordered topological space, whose topology and order are induced by those of $\Lh^D$. In fact $\Mea(D)$ is a Priestley space:
\begin{proposition}\label{p:measures-Priestley}
    For any distributive lattice $D$, $\Mea(D)$ is a Priestley space.
\end{proposition}
\begin{proof}
    It suffices to show that $\Mea(D)$ is a closed subspace of $\Lh^D$. Let
\[ 
C_{1,2} = \{f\in \Lh^D\mid f(0) = 0\cc\}\cap \{f\in \Lh^D\mid f(1) = 1\cc\} \cap \bigcap_{a\leq b} \{ f\in \Lh^D \mid f(a) \leq f(b) \}. 
\]
Note that the evaluation maps $\ev_a\colon \Lh^D\to \Lh$, $f\mapsto f(a)$, are continuous for every $a\in D$. Thus, the first set in the intersection defining $C_{1,2}$ is closed because it is the equaliser of the evaluation map $\ev_0$ and the constant map of value $0\cc$. Similarly, for the set $\{f\in \Lh^D\mid f(1) = 1\cc\}$. The last one is the intersection of the sets of the form $\langle \ev_a, \ev_b\rangle\inv (\leq)$, which are closed because $\leq$ is closed in $\Lh\times \Lh$. Whence, $C_{1,2}$ is a closed subset of $\Lh^D$. Moreover,
    \begin{align*}
        \Mea(D) = \bigcap_{a,b\in D} &\{ f\in C_{1,2} \mid f(a)\miss f(a\wedge b) \leq f(a\vee b) \mip f(b) \} \\
        \cap \bigcap_{a,b\in D} &\{ f\in C_{1,2} \mid f(a)\mip f(a\wedge b) \geq f(a\vee b) \miss f(b) \}.
    \end{align*}
    From semicontinuity of $\mip$ and $\miss$ (Lemma \ref{l:properties-of-mis}) and the following well-known fact in order-topology we conclude that $\Mea(D)$ is closed in $\Lh^D$.
    \begin{fact}
    Let $X,Y$ be compact ordered spaces, $f\colon X \to Y$ a lower semicontinuous function and $g\colon X\to Y$ an upper semicontinuous function. If $X'$ is a closed subset of $X$, then so is $E = \{ x\in X' \mid  g(x) \leq f(x)\}$.\qed
\end{fact}
\end{proof}

Next, we prove a property which is very useful when approximating a fragment of a logic by smaller fragments (see, e.g., Section \ref{s:Gamma-valued-Stone-pairing}). 
Let us denote by $\mathbf{DLat}$ the category of distributive lattices and homomorphisms, and by $\mathbf{Pries}$ the category of Priestley spaces and continuous monotone maps.
\begin{proposition}\label{p:functor-mea}
  The assignment $D\mapsto \Mea(D)$ yields a contravariant functor $\Mea\colon \mathbf{DLat}\to \mathbf{Pries}$ which sends directed colimits to codirected limits.
\end{proposition}
\begin{proof}
    If $h\colon D\to E$ is a lattice homomorphism and $\mu\colon E\to \Lh$ is a measure, it is not difficult to see that $\Mea(h)(\mu)=\mu \cdot h\colon D\to \Lh$ is a measure.
    The mapping $\Mea(h)\colon \Mea(E) \to \Mea(D)$ is clearly monotone. For continuity, recall that the topology of $\Mea(D)$ is generated by the sets $\MBD{a}{q} = \{ \nu\colon D\to \Lh \mid  \nu(a) < q\cc \}$ and $\MBU{a}{q} = \{ \nu\colon D\to \Lh \mid  \nu(a) \geq q\cc \}$,
with $a\in D$ and $q\in \Q\cap \ui$. We have
    \begin{equation*}
        \Mea(h)\inv(\MBD{a}{q})=\{ \mu\colon E\to \Lh \mid  \mu(h(a)) < q\cc \} = \MBD{h(a)}{q}
    \end{equation*}
    which is open in $\Mea(E)$. Similarly, $\Mea(h)\inv(\MBU{a}{q})=\MBU{h(a)}{q}$, showing that $\Mea(h)$ is continuous.
Thus, $\Mea$ is a contravariant functor.

The rest of the proof is a routine verification.\qed 
\end{proof}

\begin{remark}\label{rm:functor-M-cov}
  We work with the contravariant functor $\Mea\colon \mathbf{DLat}\to \mathbf{Pries}$ because $\Mea$ is concretely defined on the lattice side. However, by Priestley duality, $\mathbf{DLat}$ is dually equivalent to $\mathbf{Pries}$, so we can think of $\Mea$ as a covariant functor $\mathbf{Pries}\to\mathbf{Pries}$ (this is the perspective traditionally adopted in analysis, and also in the works of Ne{\v s}et{\v r}il and Ossona de Mendez). From this viewpoint, Section~\ref{s:logic-of-measures} provides a description of the endofunctor on $\mathbf{DLat}$ dual to $\Mea\colon \mathbf{Pries}\to\mathbf{Pries}$.
\end{remark}
Recall the maps $\gamma\colon \Lh\to \ui$ and $\ic\colon \ui\to\Lh$ from equations \eqref{eq:gamma}--\eqref{eq:ic}. In Section \ref{s:retraction-Lh-ui} we showed that this is a retraction-section pair. In Theorem \ref{th:retraction-measures} this retraction is lifted to the spaces of measures. We start with an easy observation:
\begin{lemma}\label{l:measure-maps}
Let $D$ be a distributive lattice. The following statements hold:
\begin{enumerate}
\item for every $\mu\in\Mea(D)$, $\gamma\cdot\mu\in\Meac(D)$,
\item for every $m\in\Meac(D)$, $\ic\cdot m\in\Mea(D)$.
\end{enumerate}
\end{lemma}
\begin{proof}
$1$. The only non-trivial condition to verify is finite additivity. In view of the discussion after Lemma \ref{l:gamma-ic-p-morphisms}, the map $\gamma$ preserves both minus operations on $\Lh$. Hence, for every $a,b\in D$, the inequalities 
$\mu(a) \miss \mu(a\wedge b) \leq \mu(a\vee b) \mip \mu(b)$ and
 $\mu(a) \mip \mu(a\wedge b) \geq \mu(a\vee b) \miss \mu(b)$ imply that
$\gamma\cdot\mu(a)-\gamma\cdot\mu(a\wedge b)=\gamma\cdot\mu(a\vee b)-\gamma\cdot\mu(b)$.

$2$. The first two conditions in Definition~\ref{d:measure} are immediate. The third condition follows from the fact that $\ic(r - s) = \ic(r) \mip \ic(s)$ whenever $s \leq r$ in \ui, and $x\miss y\leq x\mip y$ for every $(x,y)\in \domM$.\qed
\end{proof}

In view of the previous lemma, there are well-defined functions
\[
    \LG\colon \Mea(D) \to \Meac(D), \  \mu \mapsto \gamma\cdot \mu \ee{\text{ and }}
       \LC\colon \Meac(D) \to \Mea(D), \ m \mapsto \ic\cdot m.
\]
\begin{lemma}\label{l:gamma-continuous-monotone}
$\LG\colon \Mea(D) \to \Meac(D)$ is a continuous and monotone map.
\end{lemma}
\begin{proof}
The topology of $\Meac(D)$ is generated by the sets of the form $\{m\in \Meac(D)\mid m(a)\in O\}$, for $a\in D$ and $O$ an open subset of $\ui$. In turn,
\[
(\LG)\inv \{m\in \Meac(D)\mid m(a)\in O\}=\{\mu\in\Mea(D)\mid \mu(a)\in \gamma\inv (O)\}
\]
is open in $\Mea(D)$ because $\gamma\colon \Lh\to\ui$ is continuous by Lemma \ref{l:gamma-ic-p-morphisms}. This shows that $\LG\colon \Mea(D) \to \Meac(D)$ is continuous. 
Monotonicity is immediate.\qed
\end{proof}

Note that $\LG\colon \Mea(D) \to \Meac(D)$ is surjective, since it admits $\LC$ as a (set-theoretic) section. It follows that $\Meac(D)$ is a compact ordered space:
\begin{corollary}\label{c:class-meas-compact-pospace}
For each distributive lattice $D$, $\Meac(D)$ is a compact ordered space.
\end{corollary}
\begin{proof}
The surjection $\LG\colon \Mea(D) \to \Meac(D)$ is continuous (Lemma~\ref{l:gamma-continuous-monotone}). Since $\Mea(D)$ is compact by Proposition \ref{p:measures-Priestley}, so is $\Meac(D)$.
The order of $\Meac(D)$ is clearly closed in the product topology, thus $\Meac(D)$ is a compact ordered space.\qed
\end{proof}
Finally, we see that the set-theoretic retraction of $\Mea(D)$ onto $\Meac(D)$ lifts to the topological setting, provided we restrict to the down-set topologies. 
If $(X,\leq)$ is a partially ordered topological space, write $X^{\downarrow}$ for the space with the same underlying set as $X$ and whose topology consists of the open down-sets of $X$.
\begin{theorem}\label{th:retraction-measures}
    The maps $\LG\colon \Mea(D)^{\downarrow} \to \Meac(D)^{\downarrow}$ and $\LC\colon \Meac(D)^{\downarrow} \to \Mea(D)^{\downarrow}$ are a retraction-section pair of topological spaces.
\end{theorem}
\begin{proof}
It suffices to show that $\LG$ and $\LC$ are continuous. It is not difficult to see, using Lemma \ref{l:gamma-continuous-monotone}, that $\LG\colon \Mea(D)^{\downarrow} \to \Meac(D)^{\downarrow}$ is continuous. For the continuity of $\LC$, note that the topology of $\Mea(D)^{\downarrow}$ is generated by the sets of the form 
$\{\mu\in \Mea(D)\mid \mu(a)\leq q\mm\}$, for $a\in D$ and $q\in\Q\cap (0,1]$.
 We have 
\begin{align*}
(\LC)\inv \{\mu\in \Mea(D)\mid \mu(a)\leq q\mm\}&=\{m\in\Meac(D)\mid m(a)\in \ic\inv(\downset q\mm)\} \\
&=\{m\in\Meac(D)\mid m(a)<q\},
\end{align*}
which is an open set in $\Meac(D)^{\downarrow}$. This concludes the proof.\qed
\end{proof}

\section{The $\Lh$-valued Stone pairing and limits of finite structures}\label{s:stone-pairing}
In the work of Ne{\v s}et{\v r}il and Ossona de Mendez, the Stone pairing $\left<\ARG,A\right>$ is $\ui$-valued, i.e.\ an element of $\Meac(\FO(\sigma))$.
In this section, we show that basically the same construction for the recognisers arising from the application of a layer of semiring quantifiers in logic on words (cf.\ Section~\ref{s:duality-and-low}) provides an embedding of finite $\sigma$-structures into the space of $\Lh$-valued measures. It turns out that this embedding is a $\Lh$-valued version of the Stone pairing. Hereafter we make a notational difference, writing $\SPc{\ARG,\ARG}$ for the (classical) \ui-valued Stone pairing. 

The main ingredient of the construction are the $\Lh$-valued finitely supported functions.
To start with, we point out that the partial operation $\mip$ on $\Lh$ uniquely determines a partial ``plus'' operation on $\Lh$. 
Define
\[ 
+\colon \domP \to \Lh, \qtq{where} \domP = \{ (x,y) \mid x \leq 1\cc \mip y\} ,
\]
by the following rules (whenever the expressions make sense):
\[
        r\cc+s\cc=(r+s)\cc, \ \ 
        r\mm+s\cc=(r+s)\mm, \ \
        r\cc+s\mm=(r+s)\mm,\ete{and }
        r\mm+s\mm=(r+s)\mm.
\]
Then, for every $y\in \Lh$, $\blank + y$ is left adjoint to $\blank \mip y$.

\begin{definition}
For any set $X$, $\Fs(X)$ is the set of all functions $f\colon X\to \Lh$ s.t.\
\begin{enumerate}
    \item the set $\supp(f) = \{ x\in X \mid   f(x) \not= 0\cc\}$ is finite, and
    \item $f(x_1)+\cdots +f(x_n)$ is defined and equal to $1\cc$, where $\supp(f) = \{x_1, \ldots, x_n\}$.
\end{enumerate}
\end{definition}

To improve readability, if the sum $y_1+\cdots +y_m$ exists in $\Lh$, we denote~it $\sum_{i=1}^{m}{y_i}$. Finitely supported functions in the above sense always determine measures over the power-set algebra (the proof is an easy verification and is omitted):
\begin{lemma}\label{l:fs-lift-to-mea}
    Let $X$ be any set. There is a well-defined mapping $\int\colon\Fs(X) \to \Mea(\P(X))$, assigning to every $f\in \Fs(X)$ the measure 
    \[\textstyle
    \int f\colon M\mapsto \int_M f = \sum \{ f(x) \mid  x\in M\cap \supp(f)\}.
    \]
\end{lemma}

\subsection{The $\Lh$-valued Stone pairing and logic on words}\label{s:Gamma-valued-Stone-pairing}
Fix a countably infinite set of variables $\{ v_1, v_2, \dots \}$. Recall that $\FO_n(\sigma)$ is the Lindenbaum-Tarski algebra of first-order formulas with free variables among $\{ v_1, \dots, v_n\}$. The dual space of $\FO_n(\sigma)$ is the space of $n$-types $\Typ_n(\sigma)$. Its points are the equivalence classes of pairs $(A, \alpha)$, where $A$ is a $\sigma$-structure and $\alpha\colon \{ v_1, \dots, v_n\} \to A$ is an interpretation of the variables.
Write $\Fin(\sigma)$ for the set of all finite $\sigma$-structures. We have a map $\Fin(\sigma) \to \Fs(\Typ_n(\sigma))$ defined by $A \mapsto f_n^A$, where $f_n^A$ is the function which sends an equivalence class $E \in \Typ_n(\sigma)$ to
\[
    f_n^A(E)\ = \sum_{(A, \alpha) \in E} \left(\frac{1}{|A|^{n}}\right)\cc
    \quad \parbox{20.5em}{\centering\emph{(Add $\frac{1}{|A|^n}$ for every interpretation $\alpha$ of the free variables s.t.\ $(A, \alpha)$ is in the equivalence class)}.}
 \]
By Lemma~\ref{l:fs-lift-to-mea}, we get a measure $\int f_n^A\colon \P(\Typ_n(\sigma))\to \Lh$. Now, for each $\phi\in\FO_n(\sigma)$, let $\sem[\phi]_n \sue \Typ_n(\sigma)$ be the set of (equivalence classes of) $\sigma$-structures with interpretations satisfying $\phi$. By Stone duality we obtain an embedding $\sem_n\colon \FO_n(\sigma)\mono \P(\Typ_n(\sigma))$. Restricting $\int f_n^A$ to $\FO_n(\sigma)$, we get a measure
\[\textstyle 
\mu_n^A \colon \FO_n(\sigma)\to \Lh,\quad \phi \mapsto \int_{\sem[\phi]_n} f_n^A. 
\]
Summing up, we have the composite map
\begin{equation}\label{eq:comparison-with-logic-on-words}\textstyle
\Fin(\sigma)\to \Mea(\P(\Typ_n(\sigma)))\to \Mea(\FO_n(\sigma)), \quad A\mapsto \int f_n^A\mapsto \mu_n^A.
\end{equation}
Essentially the same construction is featured in logic on words, cf.\ equation \eqref{def:R}:
\begin{itemize}[leftmargin=12pt]
\item The set of finite $\sigma$-structures $\Fin(\sigma)$ corresponds to the set of finite words $A^*$.
\item The collection $\Typ_n(\sigma)$ of (equivalence classes of) $\sigma$-structures with interpretations corresponds to $(A\times 2)^*$ or, interchangeably, $\beta(A\times 2)^*$ (in the case of one free variable).
\item The fragment $\FO_n(\sigma)$ of first-order logic corresponds to the Boolean algebra of languages, defined by formulas with a free variable, dual to the Boolean space $X$ appearing in \eqref{def:R}.
\item The first map in the composite \eqref{eq:comparison-with-logic-on-words} sends a finite structure $A$ to the measure $\int f_n^A$ which, evaluated on $K\sue \Typ_n(\sigma)$, counts the (proportion of) interpretations $\alpha\colon \{v_1,\dots,v_n\}\to A$ such that $(A,\alpha)\in K$, similarly to $R$ from~\eqref{def:R}.
\item Finally, the second map in \eqref{eq:comparison-with-logic-on-words} sends a measure in $\Mea(\P(\Typ_n(\sigma)))$ to its pushforward along $\sem_n\colon \FO_n(\sigma)\mono \P(\Typ_n(\sigma))$. This is the second map in the composition \eqref{def:R}.
\end{itemize}

On the other hand, the assignment $A\mapsto \mu_n^A$ defined in \eqref{eq:comparison-with-logic-on-words} is also closely related to the classical Stone pairing.
Indeed, for every formula $\phi$ in $\FO_n(\sigma)$,
    \begin{align}
        \mu_n^A(\phi)
            = \sum_{(A, \alpha) \in \sem[\phi]_n} \left(\frac{1}{|A|^n}\right)\cc
            = \left( \frac{|\{ \overline a\in A^n \mid A\models \phi(\overline a)|}{|A|^n}\right)\cc = (\SPc{\phi, A})\cc. \label{eq:Stone-pairing-F}
    \end{align}
In this sense, $\mu_n^A$ can be regarded as a $\Lh$-valued Stone pairing, relative to the fragment $\FO_n(\sigma)$.
Next, we show how to extend this to the full first-order logic $\FO(\sigma)$. First, we observe that the construction is invariant under extensions of the set of free variables (the proof is the same as in the classical case).
\begin{lemma}\label{l:meas-restr}
Given $m, n\in \N$ and $A\in \Fin(\sigma)$, if $m\geq n$ then $(\mu_{m}^A)_{\restriction\FO_n(\sigma)}=\mu_n^A$.
\end{lemma}
The Lindenbaum-Tarski algebra of all first-order formulas $\FO(\sigma)$ is the directed colimit  
of the Boolean subalgebras $\FO_n(\sigma)$, for $n\in \N$. Since the functor $\Mea$ turns directed colimits into codirected limits (Proposition~\ref{p:functor-mea}), the Priestley space $\Mea(\FO(\sigma))$ is the limit of the diagram
\[ \left\{
    \begin{tikzcd}[cramped, sep=2.5em]
        \Mea(\FO_n(\sigma)) & {\Mea(\FO_m(\sigma)) \mid m,n\in\N, \ m\geq n} \ar[swap, two heads]{l}{q_{n,m}}
    \end{tikzcd}
   \right\}
\]
where, for any $\mu\colon \FO_m(\sigma) \to \Lh$ in $\Mea(\FO_m(\sigma))$, the measure $q_{n,m}(\mu)$ is the restriction of $\mu$ to $\FO_n(\sigma)$. In view of Lemma~\ref{l:meas-restr}, for every $A\in\Fin(\sigma)$, the tuple $(\mu_n^A)_{n\in\N}$ is compatible with the restriction maps. Thus, recalling that limits in the category of Priestley spaces are computed as in sets, by universality of the limit construction, this tuple yields a measure
\[ \SP{\ARG, A} \colon \FO(\sigma) \to \Lh \]
in the space $\Mea(\FO(\sigma))$. This we call the \emph{$\Lh$-valued Stone pairing} associated with $A$.
As in the classical case, it is not difficult to see that the mapping $A \mapsto \SP{\ARG, A}$ gives an embedding $\SP{\ARG,\ARG}\colon \Fin(\sigma) \mono \Mea(\FO(\sigma))$. 
The following theorem illustrates the relation between the classical Stone pairing $\SPc{\ARG,\ARG}\colon \Fin(\sigma)\mono \Meac(\FO(\sigma))$, and the $\Lh$-valued one.
\begin{theorem}\label{t:embedding-triangle}
    The following diagram commutes:
    \begin{center}
        \begin{tikzcd}[row sep=small]
            {} & \Mea(\FO(\sigma))\ar[bend left=25]{dd}{\LG} \\
            \Fin(\sigma) \ar{ru}{\SP{\ARG,\ARG}}\ar[swap]{rd}{\SPc{\ARG,\ARG}} & \\
            & \Meac(\FO(\sigma))\ar[bend left=25]{uu}{\LC}
        \end{tikzcd}
    \end{center}
\end{theorem}
\begin{proof}
    Fix an arbitrary finite structure $A\in\Fin(\sigma)$. Let $\phi$ be a formula in $\FO(\sigma)$ with free variables among $\{v_1, \dots, v_n\}$, for some $n\in \N$. By construction, $\SP{\phi,A}=\mu_n^A(\phi)$. Therefore, by equation \eqref{eq:Stone-pairing-F}, $\SP{\phi,A}=(\SPc{\phi, A})\cc$. The statement then follows at once.\qed
\end{proof}
\begin{remark*}
The construction in this section works also for proper fragments, i.e.\ for sublattices $D\subseteq \FO(\sigma)$. 
This corresponds to composing the embedding $\Fin(\sigma)\mono\Mea(\FO(\sigma))$ with the restriction map $\Mea(\FO(\sigma)) \to \Mea(D)$ sending $\mu\colon \FO(\sigma)\to\Lh$ to $\mu_{\restriction D}\colon D\to\Lh$. The only difference is that the ensuing map $\Fin(\sigma)\to \Mea(D)$ need not be injective, in general.
\end{remark*}

\subsection{Limits in the spaces of measures}\label{s:limits-finite-structures}
By Theorem \ref{t:embedding-triangle} the $\Lh$-valued Stone pairing $\SP{\ARG,\ARG}$ and the classical Stone pairing $\SPc{\ARG,\ARG}$ determine each other. However, the notions of convergence associated with the spaces $\Mea(\FO(\sigma))$ and $\Meac(\FO(\sigma))$ are different: since the topology of $\Mea(\FO(\sigma))$ is richer, there are ``fewer'' convergent sequences. 
Recall from Lemma~\ref{l:gamma-continuous-monotone} that $\LG\colon \Mea(\FO(\sigma)) \to \Meac(\FO(\sigma))$ is continuous. Also, $\LG (\SP{\ARG, A}) = \SPc{\ARG, A}$ by Theorem~\ref{t:embedding-triangle}. Thus, for any sequence of finite structures $(A_n)_{n\in\N}$, if
\begin{center}
    $\SP{\ARG, A_n}$ converges to a measure $\mu$ in $\Mea(\FO(\sigma))$
\end{center}
then
\begin{center}
    $\SPc{\ARG, A_n}$ converges to the measure $\LG(\mu)$ in $\Meac(\FO(\sigma))$.
\end{center}

The converse is not true. For example, consider the signature $\sigma=\{<\}$ consisting of a single binary relation symbol, and let $(A_n)_{n\in\N}$ be the sequence of finite posets displayed in the picture below.
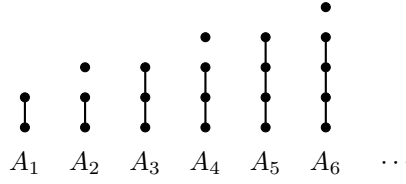
\begin{center}
\begin{tikzpicture}[xscale=0.80,yscale=0.40]
\tikzset{mynode1/.style={draw,circle,fill=black, inner sep=1.2pt,outer sep=0pt}}
    \node [mynode1, label={[yshift=-0.8cm]$A_1$}] (a1) at (0,0) {};
    \node [mynode1] (a2) at (0,1) {};

    \node [mynode1, label={[yshift=-0.8cm]$A_2$}] (b1) at (1,0) {};
    \node [mynode1] (b2) at (1,1) {};
    \node [mynode1] (b3) at (1,2) {};

    \node [mynode1, label={[yshift=-0.8cm]$A_3$}] (c1) at (2,0) {};
    \node [mynode1] (c2) at (2,1) {};
    \node [mynode1] (c3) at (2,2) {};

    \node [mynode1, label={[yshift=-0.8cm]$A_4$}] (d1) at (3,0) {};
    \node [mynode1] (d2) at (3,1) {};
    \node [mynode1] (d3) at (3,2) {};
    \node [mynode1] (d4) at (3,3) {};

    \node [mynode1, label={[yshift=-0.8cm]$A_5$}] (e1) at (4,0) {};
    \node [mynode1] (e2) at (4,1) {};
    \node [mynode1] (e3) at (4,2) {};
    \node [mynode1] (e4) at (4,3) {};

    \node [mynode1, label={[yshift=-0.8cm]$A_6$}] (f1) at (5,0) {};
    \node [mynode1] (f2) at (5,1) {};
    \node [mynode1] (f3) at (5,2) {};
    \node [mynode1] (f4) at (5,3) {};
    \node [mynode1] (f5) at (5,4) {};

    \node [label={[yshift=-0.8cm]$\cdots$}] (g1) at (6.2,0) {};

    \draw [thick] (a1) -- (a2) (b1) -- (b2) (c1) -- (c2) -- (c3) 
    (d1) -- (d2) -- (d3) (e1) -- (e2) -- (e3) -- (e4)
    (f1) -- (f2) -- (f3) -- (f4); 
\end{tikzpicture}
\end{center}
Let $\psi(x)\approx\forall y\,\neg(x<y) \wedge \exists z\,\neg(z < x)\mee \neg(z = x)$ be the formula stating that $x$ is maximal but not the maximum in the order given by $<$.
Then, for the sublattice $D=\{\false, \psi,\true\}$ of $\FO(\sigma)$, the sequences $\SP{\ARG, A_n}$ and $\SPc{\ARG, A_n}$ converge in $\Mea(D)$ and $\Meac(D)$, respectively. However, if we consider the Boolean algebra $B=\{\false, \psi,\neg\psi,\true\}$, then the $\SPc{\ARG, A_n}$'s still converge whereas the $\SP{\ARG, A_n}$'s do not. Indeed, the following sequence does not converge in $\Lh$: 
\[
(\SP{\neg\psi, A_n})_n=(1\cc, (\tfrac{1}{3})\cc, 1\cc, (\tfrac{2}{4})\cc, 1\cc, (\tfrac{3}{5})\cc,\ldots).
\]

Next, identify $\Fin(\sigma)$ with a subset of $\Mea(\FO(\sigma))$ (resp.\ $\Meac(\FO(\sigma))$) through $\SP{\ARG,\ARG}$ (resp.\ $\SPc{\ARG,\ARG}$).
A central question in the theory of structural limits, cf.\ \cite{nevsetvril2016first}, is to determine the closure of $\Fin(\sigma)$ in $\Meac(\FO(\sigma))$, which consists precisely of the limits of sequences of finite structures.  The following theorem gives an answer to this question in terms of the corresponding question for $\Mea(\FO(\sigma))$.
\begin{theorem}\label{t:closure-f-Fin}
  Let $\overline{\Fin(\sigma)}$ denote the closure of $\Fin(\sigma)$ in $\Mea(\FO(\sigma))$. Then the set $\LG(\overline{\Fin(\sigma)})$ coincides with the closure of $\Fin(\sigma)$ in $\Meac(\FO(\sigma))$.
\end{theorem}
\begin{proof}
    Write $U$ for the image of $\SP{\ARG,\ARG}\colon \Fin(\sigma) \mono \Mea(\FO(\sigma))$, and $V$ for the image of $\SPc{\ARG,\ARG}\colon \Fin(\sigma) \mono \Meac(\FO(\sigma))$. 
    We must prove that $\LG(\overline{U})=\overline{V}$. By Theorem \ref{t:embedding-triangle}, $\LG(U)=V$. The map $\LG\colon \Mea(\FO(\sigma))\to\Meac(\FO(\sigma))$ is continuous (Lemma \ref{l:gamma-continuous-monotone}), and the spaces $\Mea(\FO(\sigma))$ and $\Meac(\FO(\sigma))$ are compact Hausdorff (Proposition \ref{p:measures-Priestley} and Corollary \ref{c:class-meas-compact-pospace}). Since continuous maps between compact Hausdorff spaces are closed, $\LG(\overline{U})=\overline{\LG(U)}=\overline{V}$.
    \qed
\end{proof}

\section{The logic of measures}\label{s:logic-of-measures}
Let $D$ be a distributive lattice. We know from Proposition~\ref{p:measures-Priestley} that the space $\Mea(D)$ of $\Lh$-valued measures on $D$ is a Priestley space, whence it has a dual distributive lattice $\PFG(D)$. In this section we show that $\PFG(D)$ can be represented as the Lindenbaum-Tarski algebra for a propositional logic $\PL D$ obtained from $D$ by adding probabilistic quantifiers. 
Since we adopt a logical perspective, we write $\false$ and $\true$ for the bottom and top elements of $D$, respectively.

The set of propositional variables of $\PL D$ consists of the symbols $\PrG p a$, for every $a\in D$ and $p\in\Q\cap \ui$. For every measure $\mu\in\Mea(D)$, we set
\begin{equation}\label{eq:semantic-meas}
\mu\models \PrG p a \ \Leftrightarrow \ \mu(a)\geq p\cc.
\end{equation}
This satisfaction relation extends in the obvious way to the closure under finite conjunctions and finite disjunctions of the set of propositional variables. Define
\[ \phi \models \psi \qtq{if,} \forall \mu\in\Mea(D),\ \ \mu\models \phi \text{ implies }\mu\models \psi.\]
Also, write $\models \phi$ if $\mu\models \phi$ for every $\mu\in\Mea(D)$, and $\phi \models$ if there is no $\mu\in\Mea(D)$ with $\mu\models \phi$.

\smallskip
Consider the following conditions, for any $p,q,r\in \Q\cap \ui$ and $a,b\in D$.

\medskip
\quad
\begin{minipage}{0.91\textwidth}
    \begin{enumerate}[label=(L\arabic*)]
        \item $\PrG q a \models \PrG p a$ whenever $p \leq q$
        \item $\PrG p \false\models{}$ whenever $p>0$, ${}\models \PrG 0 \false$ and ${}\models \PrG q \true$
        \item $\PrG q a \models \PrG q b$ whenever $a\leq b$
        \item $\PrG p a\,\wedge\, \PrG q b \models \PrG{p+q-r}(a\vee b) \,\vee\, \PrG{r}(a\wedge b)$ whenever $0\leq p+q-r\leq 1$
        \item $\PrG{p+q-r}(a\vee b)\,\wedge\,\PrG{r}(a\wedge b) \models \PrG p a \,\vee\, \PrG q b$ whenever $0\leq p+q-r\leq 1$
    \end{enumerate}
\end{minipage}
\medskip

\noindent
It is not hard to see that the interpretation in \eqref{eq:semantic-meas} validates these conditions:

    \begin{lemma}\label{l:soundness-L1-L5}
    The conditions (L1)--(L5) are satisfied in $\Mea(D)$.
    \end{lemma}

    Write $\PFG(D)$ for the quotient of the free distributive lattice on the set 
    \[
    \{\PrG p a\mid p\in\Q\cap \ui, \ a\in D\}
    \]
    with respect to the congruence generated by the conditions (L1)--(L5).

\begin{proposition}\label{l:homo-to-measure}
    Let $\Pf\sue \PFG(D)$ be a prime filter. The assignment
    \[ 
    a\mapsto \bigvee{\{q\cc \mid \PrG q a \in \Pf\}} \quad \text{ defines a measure } \ \mu_{\Pf}\colon D \to \Lh.
    \]
\end{proposition}
\begin{proof}
    Items (L2) and (L3) take care of the first two conditions defining $\Lh$-valued measures (cf.\ Definition~\ref{d:measure}). We prove the first half of the third condition, as the other half is proved in a similar fashion. We must show that, for every $a,b\in D$,  
\begin{align}
    \label{eq:mu-h-1}
    \mu_{\Pf}(a)\miss \mu_{\Pf}(a\wedge b)&\leq \mu_{\Pf}(a\vee b)\mip \mu_{\Pf}(b).
\end{align}
    It is not hard to show that $\mu_{\Pf}(a) \mip r\cc = \bigvee \{ p_1\cc \mip r\cc \mid r\cc \leq p_1\cc \leq \mu_{\Pf}(a) \}$, and $x\mip \blank$ transforms non-empty joins into meets (this follows by Scott continuity of $x\mip\blank$ seen as a map $[0\cc,x]\to \Lh^\op$).
    Hence, equation \eqref{eq:mu-h-1} is equivalent to
    \begin{equation*}
    \bigvee{\{p\cc\mip r\cc\mid \mu_{\Pf}(a\wedge b)<r\cc\leq p\cc\leq\mu_{\Pf}(a)\}}\leq \bigwedge{\{\mu_{\Pf}(a\vee b)\mip q\cc\mid q\cc \leq \mu_{\Pf}(b)\}}.
    \end{equation*}
To settle this inequality it is enough to show that, provided $\mu_{\Pf}(a\wedge b)<r\cc\leq p\cc\leq \mu_{\Pf}(a)$ and $q\cc \leq \mu_{\Pf}(b)$, we have $(p-r)\cc \leq \mu_{\Pf}(a\vee b)\mip q\cc$.
The latter inequality is equivalent to $(p+q-r)\cc\leq \mu_{\Pf}(a\vee b)$.  
In turn, using (L4) and the fact that $\Pf$ is a prime filter, $\PrG p a,\PrG q b\in \Pf$ and $\PrG r (a\wedge b)\notin \Pf$ entail $\PrG{p+q-r}(a\vee b)\in \Pf$. Whence,
\[
\mu_{\Pf}(a\vee b)=\bigvee{\{s\cc \mid \PrG s (a\vee b)\in \Pf\}}\geq (p+q-r)\cc. \qedhere
\]
\end{proof}

We can now describe the dual lattice of $\Mea(D)$ as the Lindenbaum-Tarski algebra for the logic $\PL D$, built from the propositional variables $\PrG p a$ by imposing the laws (L1)--(L5). 
\begin{theorem}\label{t:PFG(D)-dual-to-meas}
   Let $D$ be a distributive lattice. Then the lattice $\PFG(D)$ is isomorphic to the distributive lattice dual to the Priestley space $\Mea(D)$.
\end{theorem}
\begin{proof}
Let $X_{\PFG(D)}$ be the space dual to $\PFG(D)$. By Proposition \ref{l:homo-to-measure} there is a map $\theta\colon X_{\PFG(D)}\to \Mea(D)$, $\Pf\mapsto \mu_{\Pf}$. We claim that $\theta$ is an isomorphism of Priestley space.     
Clearly, $\theta$ is monotone. If $\mu_{\Pf_1}(a)\not\leq \mu_{\Pf_2}(a)$ for some $a\in D$, we have
   \begin{equation}\label{eq:inequality-embed}
   \bigvee{\{q\cc\mid \PrG q a\in \Pf_1\}}=\mu_{\Pf_1}(a)\not\leq \mu_{\Pf_2}(a) = \bigwedge{\{p\mm\mid \PrG p a\notin \Pf_2\}}.
   \end{equation}
   Equation \eqref{eq:inequality-embed} implies the existence of $p,q$ satisfying $\PrG q a\in \Pf_1$, $\PrG p a\notin \Pf_2$ and $q\geq p$. It follows by (L1) that  $\PrG p a\in \Pf_1$. We conclude that $\PrG p a \in \Pf_1 \setminus \Pf_2$, whence $\Pf_1\not\sue \Pf_2$.  This shows that $\theta$ is an order embedding, whence injective. 
   
   We prove that $\theta$ is also surjective, thus a bijection. Fix a measure $\mu\in\Mea(D)$. It is not hard to see, using Lemma \ref{l:soundness-L1-L5}, that the filter $\Pf_{\mu}\sue \PFG(D)$ generated by
   \[
       \{ \PrG q a \mid a\in D,\ q\in\Q\cap\ui,\ \mu(a)\geq q\cc\}
   \]
   is prime. Further,  
    $\theta(\Pf_{\mu})(a)=\bigvee{\{q\cc\mid \PrG q a \in \Pf_\mu\}}=\bigvee{\{q\cc\mid \mu(a) \geq q\cc \}} =\mu(a)$ for every $a\in D$. Hence, $\theta(\Pf_{\mu})=\mu$ and $\theta$ is surjective. 
    
    To settle the theorem it remains to show that $\theta$ is continuous.
Note that for a basic clopen of the form $C=\{\mu\in\Mea(D)\mid \mu(a)\geq p\cc\}$ where $a\in D$ and $p\in\Q\cap \ui$, the preimage $\theta\inv(C) = \{\Pf\sue \PFG(D) \mid \mu_{\Pf}(a)\geq p\cc\}$ is equal to
   \begin{align*}
   \{\Pf\in X_{\PFG(D)}\mid \bigvee{\{q\cc\mid \PrG q a\in \Pf\}}\geq p\cc\}
   =\{\Pf\in X_{\PFG(D)}\mid \PrG p a \in \Pf\},
   \end{align*}
    which is a clopen of $X_{\PFG(D)}$.
     Similarly, if $C = \{\mu\in\Mea(D)\mid \mu(a)\leq q\mm\}$ for some $a\in D$ and $q\in\Q\cap (0,1]$, by the claim above $\theta\inv(C)=\{\Pf\in X_{\PFG(D)}\mid \PrG q a\notin \Pf\}$, which is again a clopen of $X_{\PFG(D)}$.\qed
\end{proof}

By Theorem~\ref{t:PFG(D)-dual-to-meas}, for any distributive lattice $D$, the lattice of clopen up-sets of $\Mea(D)$ is isomorphic to the Lindenbaum-Tarski algebra $\PFG(D)$ of our \emph{positive} propositional logic $\PL D$. Moving from the lattice of clopen up-sets to the Boolean algebra of all clopens logically corresponds to adding negation to the logic.
The logic obtained this way can be presented as follows.
Introduce a new propositional variable $\PrL q a$, for each $a\in D$ and $q\in \Q\cap\ui$. For a measure $\mu\in\Mea(D)$, set
\[
\mu\models \PrL q a \ \Leftrightarrow \ \mu(a)< q\cc.
\]
We also add a new rule, stating that $\PrL q a$ is the negation of $\PrG q a$:

\medskip
\quad
\begin{minipage}{0.8\textwidth}
    \begin{enumerate}[label=(L\arabic*)]\setcounter{enumi}{5}
        \item $\PrL q a \mee \PrG q a\models {}$ and ${} \models \PrL q a \vee \PrG q a$
    \end{enumerate}
\end{minipage}
\medskip

\noindent
Clearly, (L6) is satisfied in $\Mea(D)$. Moreover, the Boolean algebra of \emph{all} clopens of $\Mea(D)$ is isomorphic to the quotient of the free distributive lattice on
\[
   \{\PrG p a\mid p\in\Q\cap \ui, \ a\in D\} \cup \{\PrL q b\mid q\in\Q\cap \ui, \ b\in D\}
\]
with respect to the congruence generated by the conditions (L1)--(L6).

\paragraph*{Specialising to $\FO(\sigma)$.}
Let us briefly discuss what happens when we instantiate $D$ with the full first-order logic $\FO(\sigma)$. For a formula $\phi\in \FO(\sigma)$ with free variables $v_1, \dots, v_n$ and a $q\in \Q\cap\ui$, we have two new sentences $\PrG q \phi$ and $\PrL q \phi$. For a finite $\sigma$-structure $A$ identified with its $\Lh$-valued Stone pairing $\SP{\ARG,A}$,
\[ A \models \PrG q \phi \ \text{ (resp.\ } A \models \PrL q \phi) \qtq{iff} \SP{\phi, A} \geq q\cc \ \text{ (resp.\ } \SP{\phi, A} < q\cc). \]
That is, $\PrG q \phi$ is true in $A$ if a random assignment of the variables $v_1, \dots, v_n$ in $A$ satisfies $\phi$ with probability at least $q$. Similarly for $\PrL q \phi$. If we regard $\PrG q$ and $\PrL q$ as probabilistic quantifiers that bind all free variables of a given formula, the Stone pairing $\SP{\ARG,\ARG}\colon \Fin \to \Mea(\FO(\sigma))$ can be seen as the embedding of finite structures into the space of types for the logic $\PL{\FO(\sigma)}$.

\section*{Conclusion}
Types are points of the dual space of a logic (viewed as a Boolean algebra). In classical first-order logic, $0$-types are just the models modulo elementary equivalence. But when there are not `enough' models, as in finite model theory, the spaces of types provide completions of the sets of models.

In \cite{GPR2017}, it was shown that for logic on words and various quantifiers we have that, given a Boolean algebra of formulas with a free variable, the space of types of the Boolean algebra generated by the formulas obtained by quantification is given by a measure space construction. Here we have shown that a suitable enrichment of first-order logic gives rise to a space of measures $\Mea(\FO(\sigma))$ closely related to the space $\Meac(\FO(\sigma))$ used in the theory of structural limits. Indeed Theorem~\ref{t:embedding-triangle} tells us that the ensuing Stone pairings interdetermine each other. Further, the Stone pairing for $\Mea(\FO(\sigma))$  is just the embedding of the finite models in the completion/compac\-tification provided by the space of types of the enriched logic.

These results identify the logical gist of the theory of structural limits, and provide a new and interesting connection between logic on words and the theory of structural limits in finite model theory. But we also expect that it may prove a useful tool in its own right. 
Thus, for structural limits, it is an open problem to characterise the closure of the image of the $[0,1]$-valued Stone pairing  \cite{nevsetvril2016first}. Reasoning instead in the $\Lh$-valued setting, native to logic and where we can use duality, one would expect that this is the subspace $\Mea({\rm Th}(\Fin))$ of  $\Mea(\FO(\sigma))$ given by the quotient $\FO(\sigma)\twoheadrightarrow{\rm Th}(\Fin)$ onto the theory of pseudofinite structures. The purpose of such a characterisation would be to understand the points of the closure as  ``generalised models". 
Another subject that we would like to investigate is that of zero-one laws. The zero-one law for first-order logic states that the sequence of measures for which the $n$th measure, on a sentence $\psi$, yields the proportion of $n$-element structures satisfying $\psi$, converges to a $\{0,1\}$-valued measure. Over $\Lh$ this will no longer be true as $1$ is split into its `limiting' and `achieved' personae.  
Yet, we expect the above sequence to converge also in this setting and by Theorem~\ref{t:embedding-triangle}, it will converge to a $\{0\cc,1\mm,1\cc\}$-valued measure.  
Understanding this more fine-grained measure may yield useful information about the zero-one law.

Further, it would be interesting to investigate whether the limits for schema mappings introduced by Kolaitis \emph{et al.} \cite{Kolaitis2018} may be seen also as a type-theoretic construction. Finally, we would want to explore the connections with other semantically inspired approaches to finite model theory, such as those recently put forward by Abramsky, Dawar \emph{et al.} \cite{Abramsky2017b,AbramskyShah2018}.


\begin{thebibliography}{10}
\providecommand{\url}[1]{\texttt{#1}}
\providecommand{\urlprefix}{URL }
\providecommand{\doi}[1]{https://doi.org/#1}

\bibitem{Abramsky91}
Abramsky, S.: Domain theory in logical form. Ann. Pure Appl. Logic
  \textbf{51},  1--77 (1991)

\bibitem{Abramsky2017b}
Abramsky, S., Dawar, A., Wang, P.: The pebbling comonad in finite model theory.
  In: 32nd Annual {ACM/IEEE} Symposium on Logic in Computer Science, {LICS}.
  pp. 1--12 (2017)

\bibitem{AbramskyShah2018}
Abramsky, S., Shah, N.: Relating {S}tructure and {P}ower: Comonadic semantics
  for computational resources. In: 27th {EACSL} Annual Conference on Computer
  Science Logic, {CSL}. pp. 2:1--2:17 (2018)

\bibitem{GGP2008}
Gehrke, M., Grigorieff, S., Pin, J.{\'E}.: Duality and equational theory of
  regular languages. In: Automata, languages and programming {II}, LNCS,
  vol.~5126, pp. 246--257. Springer, Berlin (2008)

\bibitem{GPR2017}
Gehrke, M., Petri\c{s}an, D., Reggio, L.: Quantifiers on languages and
  codensity monads. In: 32nd Annual {ACM/IEEE} Symposium on Logic in Computer
  Science, {LICS}. pp. 1--12 (2017)

\bibitem{GPR2019}
Gehrke, M., Petri\c{s}an, D., Reggio, L.: Quantifiers on languages and
  codensity monads (2019), extended version. Submitted. Preprint available at
  \href{https://arxiv.org/abs/1702.08841}{\url{https://arxiv.org/abs/1702.08841}}

\bibitem{GP2007}
Gehrke, M., Priestley, H.A.: Canonical extensions of double quasioperator
  algebras: an algebraic perspective on duality for certain algebras with
  binary operations. J. Pure Appl. Algebra  \textbf{209}(1),  269--290 (2007)

\bibitem{GP2007b}
Gehrke, M., Priestley, H.A.: Duality for double quasioperator algebras via
  their canonical extensions. Studia Logica  \textbf{86}(1),  31--68 (2007)

\bibitem{Goldblatt1989}
Goldblatt, R.: Varieties of complex algebras. Ann. Pure Appl. Logic
  \textbf{44}(3),  173--242 (1989)

\bibitem{vanGoolSteinberg}
Gool, S.J.v., Steinberg, B.: Pro-aperiodic monoids via saturated models. In:
  34th Symposium on Theoretical Aspects of Computer Science, {STACS}. pp.
  39:1--39:14 (2017)

\bibitem{Johnstone1986}
Johnstone, P.T.: Stone spaces, Cambridge Studies in Advanced Mathematics,
  vol.~3. Cambridge University Press (1986), reprint of the 1982 edition

\bibitem{Jung13}
Jung, A.: Continuous domain theory in logical form. In: Coecke, B., Ong, L.,
  Panangaden, P. (eds.) Computation, Logic, Games, and Quantum Foundations,
  Lecture Notes in Computer Science, vol.~7860, pp. 166--177. Springer Verlag
  (2013)

\bibitem{Kolaitis2018}
Kolaitis, P.G., Pichler, R., Sallinger, E., Savenkov, V.: Limits of schema
  mappings. Theory of Computing Systems  \textbf{62}(4),  899--940 (2018)

\bibitem{MaSch08}
Matz, O., Schweikardt, N.: Expressive power of monadic logics on words, trees,
  pictures, and graphs. In: Logic and Automata: History and Perspectives. pp.
  531--552 (2008)

\bibitem{nevsetvril2016first}
Ne{\v s}et{\v r}il, J., Ossona~de Mendez, P.: First-order limits, an analytical
  perspective. European Journal of Combinatorics  \textbf{52},  368--388 (2016)

\bibitem{NO2017}
Ne{\v s}et{\v r}il, J., Ossona~de Mendez, P.: A unified approach to structural
  limits and limits of graphs with bounded tree-depth (2020), to appear in
  \emph{Memoirs of the American Mathematical Society}

\bibitem{nevsetril2012model}
Ne{\v{s}}etril, J., de~Mendez, P.O.: A model theory approach to structural
  limits. Commentationes Mathematicae Universitatis Carolinae  \textbf{53}(4),
  581--603 (2012)

\bibitem{Pin09}
Pin, J.E.: {Profinite Methods in Automata Theory}. In: Albers, S., Marion, J.Y.
  (eds.) {26th International Symposium on Theoretical Aspects of Computer
  Science STACS 2009}. pp. 31--50 (2009)

\bibitem{Priestley1970}
Priestley, H.A.: Representation of distributive lattices by means of ordered
  stone spaces. Bull. London Math. Soc.  \textbf{2},  186--190 (1970)

\end{thebibliography}
\end{document}